\documentclass[a4paper,superscriptaddress,11pt,accepted=2022-02-10]{quantumarticle} \pdfoutput=1
\usepackage{cite}
\usepackage[utf8]{inputenc}
\usepackage[T1]{fontenc}
\usepackage[UKenglish]{babel}
\usepackage[colorlinks]{hyperref}
\usepackage{graphicx}
\usepackage{amsmath,amssymb,amsthm}
\usepackage{xspace}
\usepackage{mathtools}
\usepackage{dsfont}

\DeclareMathOperator{\tr}{tr}

\newcommand{\id}{\mathds{1}}

\newcommand{\ee}{\mathrm{e}}             
\newcommand{\ii}{\mathrm{i}}             

\newcommand{\bra}[1]{\left\langle #1 \right|}
\newcommand{\ket}[1]{\left| #1 \right\rangle}

\newcommand{\ketbra}[2]{\left|#1\middle\rangle\middle\langle#2\right|}
\newcommand{\proj}[1]{\left[#1\right]}

\newcommand{\Bra}[1]{{ \langle \! \langle{#1}\vert }}
\newcommand{\Ket}[1]{{ \vert {#1}  \rangle \!  \rangle}}
\newcommand{\KetBra}[1]{{\Ket{#1}\!\Bra{#1} }}

\newcommand{\Proj}[1]{[[#1]]}

\newcommand{\reftoeq}[1]{Eq.~\eqref{#1}}

\newtheorem{theorem}{Theorem}

\newtheorem{lemma}[theorem]{Lemma}
\newtheorem*{lemma*}{Lemma}
\newtheorem{definition}[theorem]{Definition}

\begin{document}

\title{A no-go theorem for superpositions of causal orders}
\author{Fabio Costa}
\email{f.costa@uq.edu.au}
\affiliation{Centre for Engineered Quantum Systems, School of Mathematics and Physics, The University of Queensland, St Lucia, QLD 4072, Australia}
\begin{abstract}
The causal order of events need not be fixed: whether a bus arrives before or after another at a certain stop can depend on other variables---like traffic. Coherent quantum control of causal order is possible too and is a useful resource for several tasks. However, quantum control implies that a controlling system carries the which-order information---if the control is traced out, the order of events remains in a probabilistic mixture. Can the order of two events be in a pure superposition, uncorrelated with any other system? Here we show that this is not possible for a broad class of processes: a pure superposition of any pair of Markovian, unitary processes with equal local dimensions and different causal orders is not a valid process, namely it results in non-normalised probabilities when probed with certain operations. The result imposes constraints on novel resources for quantum information processing and on possible processes in a theory of quantum gravity.
\end{abstract}
\maketitle

\section{Introduction}
Quantum superpositions can be viewed as generalisations of classical probabilities: if classically we can be uncertain between two alternatives, assigning to each a probability, quantumly we should be able to consider a superposition of the two, replacing probabilities with complex amplitudes. Feynman's sum-over-histories approach famously leverages this intuition \cite{Feynman1948}.
Despite this view, one typically considers superpositions of \emph{states}, whereas classical probabilities can be assigned to \emph{any} logical statement. To what extent is it possible to generalise the superposition principle, beyond its original range of applicability? Providing a general, principled answer to this question appears problematic because, unlike classical probabilities, quantum superpositions have no clear interpretation in terms of subjective lack of knowledge.

Of particular interest is the case of causal relations. A main motivation comes from combining the principles of quantum theory and general relativity. In a regime where spacetime itself becomes subject to the laws of quantum mechanics, we expect that causal structure, too, will become nonclassical~\cite{Butterfield2001, hardy2007towards, Zych2019}. A second, more practical, motivation is that quantum causal structures can also be realised experimentally (on a classical spacetime background)~\cite{Procopio2015, Rubinoe1602589, rubino2017experimental, Goswami2018, goswami2018communicating, guo2018experimental, Wei2019, taddei2020experimental} possibly leading to an advantage in solving computational and information-theoretic tasks \cite{chiribella09, chiribella12, colnaghi11, araujo14, feixquantum2015, Guerin2016, Ebler2018, Salek2018, Gupta2019}.

\begin{figure}[t]%
\begin{center}
\includegraphics[width=0.9\columnwidth]{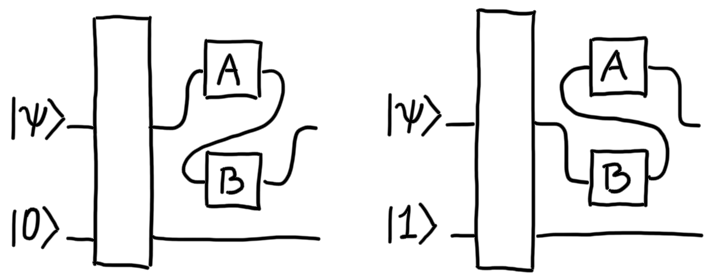}%
\end{center}
\caption{\textbf{The quantum switch.} The state of a control system determines the order of two operations on a target system. 
}%
\label{switch}%
\end{figure}

Despite much recent research, the scenarios considered so far do not support a direct interpretation as superpositions of causal orders. The most discussed example is the so-called ``quantum switch''~\cite{chiribella09}, Fig.~\ref{switch},  where a control system determines the order in which a set of operations act on a target. Preparing the control in a superposition produces a nonclassical causal structure for the order of operations. Although sometimes colloquially described as a ``superposition of causal orders'', the quantum switch is in fact more appropriately interpreted as `entanglement' between causal relations and the control system. Indeed, if the control is discarded, the quantum switch is indistinguishable from a probabilistic mixture of causal orders. Only by measuring the control can one verify the nonclassicality of the order of operations. Probabilistic protocols, where one obtains a superposition of orders through postselection, generally succeed with probability less than one, making them not commensurable with deterministic resources.

\begin{figure}%
\includegraphics[width=0.9\columnwidth]{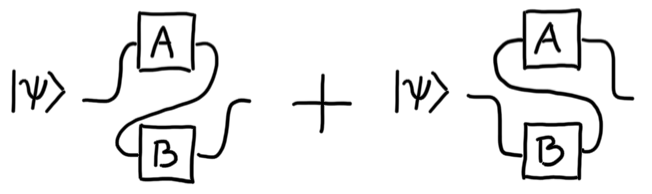}%
\caption{\textbf{Pure superposition of causal orders.} Is it possible to superpose the causal order of two events, without the aid of a control system?}%
\label{puresuper}%
\end{figure}

Is a `pure', deterministic superposition of causal orders possible (Fig.~\ref{puresuper}), not relying on any additional control system? Perhaps, this would provide a novel resource, enabling tasks not achievable by the quantum switch. It might also suggest types of processes arising in a theory of quantum gravity.

Here, we propose a definition of superpositions of causal orders and present a negative result, ruling out the most natural candidates.
We show that for unitary, Markovian processes where all systems have equal dimension superpositions of two different causal orders are not possible, in the sense that they do not constitute valid, deterministic processes. In particular, the result applies to simple sequences of operations on a target system---as in the switch---showing that an additional system is necessary for performing such operations in an indefinite causal order. It remains an open question whether a superposition is possible under weaker conditions.

\section{Superpositions of unitaries}
Before considering superpositions of different causal orders, it is useful to look at a simpler case: superpositions of unitary transformations connecting an event in the past---a state preparation---to one in the future---a measurement. Given two unitaries $U_j$, $j=0,1$, it is natural to define their superposition as a linear combination $\alpha U_0 + \beta U_1$ for some complex numbers $\alpha$, $\beta$, interpreted as probability amplitudes. However, such a `superposition' is not necessarily a unitary operator. 
For example, no linear combination of $U_0=\id$ and $U_1=\frac{\sigma_x + \ii \sigma_y}{\sqrt{2}}$ (where $\sigma_x$, $\sigma_y$, $\sigma_z$ denote the Pauli matrices) is unitary, for \emph{any} complex amplitudes $\alpha\neq 0$, $\beta\neq 0$.

One could interpret arbitrary linear combinations of unitaries in the following way \cite{Aharonov1990}: Introduce the `controlled unitary' $U_0\otimes \proj{0} + U_1\otimes \proj{1}$, which acts on an additional `control' system and where we use the shorthand $\proj{\psi}=\ketbra{\psi}{\psi}$. Starting from a state of the form $\ket{\psi}\left(\alpha\ket{0}+\beta\ket{1}\right)$, and postselecting the control on $\left(\bra{0}+\bra{1}\right)/\sqrt{2}$, the target system is left into $\left(\alpha U_0 + \beta U_1\right)\ket{\psi}/\sqrt{2}$, which one can interpret as resulting from a superposition of unitaries. The problem with this interpretation is that the postselection step requires obtaining one out of a set of possible measurement outcomes, which in general only succeeds with some probability that depends on the initial state. As such a probability can be very small, or even vanish, realising a process through postselection can nullify the associated advantages. Furthermore, postselection can induce apparent signalling, making postselected processes unsuitable for studies of causal relations\footnote{For example, one can place bets on a roulette and consider only the winning spins, effectively predicting the roulette's outcome in each of the postselected spins. This is unlikely to be a remunerative strategy in any casino.}. For these reasons, here we are only interested in processes that can be realised \emph{without} postselection.

We see that, for processes, superposition is not a universal possibility, as it is for states. However, there are unitaries for which superpositions \emph{are} possible. For example, the \emph{Hadamard gate} $H:=\frac{\sigma_x+  \sigma_z}{\sqrt{2}}$ is a rightful superposition of $\sigma_x$ and $\sigma_z$, not requiring any postselection. $H$ is a ubiquitous resource in quantum information and computation, precisely for its ability to convert basis states into superpositions. Therefore, even if we do not expect to be able to superpose arbitrary processes with different causal orders, we can still ask whether \emph{any} such a superposition is possible.

\section{General quantum processes}\label{processsection}
As we will see, superpositions of causal orders become meaningful for processes connecting more than two events. These can be described conveniently in the so-called ``process matrix formalism'' \cite{oreshkov12} (see also other closely related frameworks \cite{Oeckl2003318, Aharonov2009, Cotler2017, Silva2017, Barrett2019}). A simple example of a multi-event process is one where a sequence of unitaries $U_1,U_2,\dots$ connects a set of events at times $t_1,t_2,\dots$  At each time $t_j$, after unitary $U_{j-1}$ and before unitary $U_j$, it is in principle possible to perform arbitrary measurements or transformations on the system. The particular operation taking place at a given time is what we call an `event', while the set of unitaries connecting the events is the `process'. More generally, we consider a scenario where the labels that identify possible events are not necessarily associated with time instants (for example, events might be delocalised in time \cite{Guerin2018, Oreshkov2019timedelocalized} or take place on a nonclassical background \cite{Zych2019}). A single `event location' with label $A$ (also referred to as ``region'', ``party'', or ``laboratory'') is therefore identified with the set of possible operations from an input to an output Hilbert space, $\mathcal{H}^A_I$ and $\mathcal{H}^A_O$, respectively. 

The most general operation \cite{Heinosaari2011} is described by a Completely Positive (CP) trace non-increasing map $\mathcal{M}^A: A_I\rightarrow A_O$, where $A_{I,O} \equiv \mathcal{L}(\mathcal{H}^A_{I,O})$ are the spaces of linear operators over the respective Hilbert spaces. A deterministic operation, i.e., one that happens with probability one, is CP and trace preserving (CPTP): $\tr\circ\mathcal{M}^A= \tr$, where $\tr$ and $\circ$ denote operator trace and function composition, respectively.
A collection of CP maps $\{\mathcal{M}^A_a\}_a$, where $a$ labels the measurement outcome, is called an \emph{instrument} if $\sum_a \mathcal{M}^A_a$ is CPTP, which implies that there is unit probability that at least one of the outcomes will occur. 
It is convenient to represent CP maps as operators, using a version of the Choi-Jamio{\l}kowksi isomorphism \cite{Choi1975, jamio72}, $\mathcal{M}^A\mapsto M^A \in A_I\otimes A_O$:
\begin{equation}
M^A := \left[\sum_{jk} \ketbra{j}{k}\otimes \mathcal{M}^A\left(\ketbra{j}{k}\right)\right]^T,
\label{Choi}
\end{equation}
where $\{\ket{j}\}_j$ is a basis of $\mathcal{H}^A_I$ and $^T$ denotes transposition.

The local validity of quantum mechanics implies that the joint probability for outcomes $a_1, \dots, a_n$ in regions $A^1,\dots,A^n$ is given by \cite{oreshkov12, Shrapnel2017}
\begin{equation}
P(a_1, \dots, a_n) = \tr\left[ W \cdot \left(M_{a_1}^{A^1}\otimes \dots\otimes M_{a_n}^{A_n}  \right)\right],
\label{process matrix}
\end{equation}
where $W\in A^1 \otimes \dots \otimes A^n$ is called the \emph{process matrix} (or simply \emph{process}) and we use the shorthand $A^j\equiv A^j_I\otimes A^j_O$. $W$ encodes all information relevant to the possible events, such as initial state, transformations, and causal relations. Positivity of probabilities (and the possibility to extend the local operations to additional systems) implies that the process matrix is positive semidefinite, $W\geq 0$, while normalisation requires
\begin{equation}
\tr\left[ W \cdot \left(\bar{M}^{A^1}\otimes \dots \otimes\bar{M}^{A^n} \right)\right] = 1
\label{constraints}
\end{equation}
for all CPTP maps $\bar{M}^{A^1}, \dots, \bar{M}^{A^n}$. Crucially, Eq.~\eqref{constraints} imposes on $W$ more constraints than the normalisation condition for states, $\tr \rho =1$. Therefore, whereas all positive semidefinite operators represent states (up to normalisation) the same is not true for process matrices. As we will see, it is this constraint that obstructs general superpositions of processes. More general processes, that do not satisfy the linear constraints implied by Eq.~\eqref{constraints}, can always be realised probabilistically, through postselection \cite{Aharonov2009, oreshkov2016, Silva2017, Milz2018}. However, as discussed earlier, we are only interested in processes that can be realised without postslection. Therefore, in the following, we will always assume processes to be deterministic (and condition \eqref{constraints} to be satisfied).

The key property of the formalism is that it encodes signalling correlations: a choice of instrument in a region can change the marginal probability for measurement outcomes in another region. We say that a process is \emph{causally ordered} if there is a permutation $\sigma$ of $n$ elements such that $A^{\sigma(j)}$ cannot signal to any group of parties $\left\{A^{\sigma(k)}\right\}_k$ with $\sigma(k)<\sigma(j)$.  Conversely, we say that $W$ is \emph{incompatible} with (the causal order induced by) $\sigma$ if it possible to signal from some $A^{\sigma(j)}$ to a group of parties $\left\{A^{\sigma(k)}\right\}_k$, all with $\sigma(k)<\sigma(j)$. 
We say that two processes $W_1$, $W_2$ are \emph{differently ordered} if they are causally ordered relative to permutations $\sigma_1$, $\sigma_2$, but incompatible with $\sigma_2$, $\sigma_1$, respectively. We are interested in superpositions of differently ordered processes.

A related concept is that of \emph{causally nonseparable process}, which formalises the notion of indefinite causal order \cite{oreshkov12, oreshkov15, Wechs2019}. A process is causally separable if it is causally ordered or a probabilistic mixture of causally ordered processes. A superposition of differently ordered processes would necessarily be causally nonseparable (because it would be extremal and not causally ordered). However, the existence of locally classical, causally nonseparable processes \cite{baumeler14, baumeler2017, Baumeler2019, Tobar2020} indicates that the opposite is not true, making causal nonseparability unsuitable to characterise superpositions of causal orders.

\section{Superpositions of pure processes}
Process matrices are a generalisation of density matrices and, as such, cannot be directly superposed. In analogy to states, we can define superpositions for `pure' processes, namely, for rank-1 process matrices, with
$W = \proj{\omega}$
for some \emph{process vector} $\ket{\omega}\in \mathcal{H}^{A^1}_I \otimes\mathcal{H}^{A^1}_O\otimes\dots$ Given two pure process $\proj{\omega_1}$, $\proj{\omega_2}$, we define their superposition through a linear combination $\ket{\omega} = \alpha \ket{\omega_1} + \beta \ket{\omega_2}$, with $\alpha, \beta \in\mathbb{C}$, \emph{if $\proj{\omega}$ is a valid process matrix}. That is to say, if and only if
\begin{align}
\bra{\omega} &\left(\bar{M}^{A^1}\otimes \dots \otimes\bar{M}^{A^n}\right)\ket{\omega} = 1 
\label{pureconstraints}
\\ \nonumber
\forall &\textrm{ CPTP } \bar{M}^{A^1}, \dots , \bar{M}^{A^n}.
\end{align}

If all parties have only input spaces (or, equivalently, if all output spaces are trivial, namely they have dimension 1), the process matrix reduces to a state, and superposition of processes reduces to superposition of states. Another special case is a process that connects a region's output, $A_O$, to another region's input, $B_I$, of equal dimension. In this case, a pure process represents a unitary transformation (see, e.g., Appendix A of Ref.~\cite{araujo15}), for which we use the notation
\begin{equation}
\Ket{U}^{A_OB_I}:=\sum_j \ket{j}^{A_O}\otimes\left( U\ket{j}\right)^{B_I}
\label{doubleket}
\end{equation} 
(where $\{\ket{j}\}_j$ is the same basis that defines the Choi-Jamio{\l}kowski isomorphism).
Because of the linearity of the representation \eqref{doubleket}, a `superposition of unitary processes' reduces to the `superposition of unitaries' discussed earlier: $\alpha\Ket{U_1} + \beta\Ket{U_2} = \Ket{\alpha U_1+\beta U_2}$.

More generally, a pure process defines an \emph{isometry} from all the output to all the input spaces: if all parties prepare a pure state, they all receive a pure state, related to the prepared ones through an isometry.
In the particular case where the total input and output dimensions are equal, the process defines a unitary transformation from outputs to inputs \cite{Araujo2017purification}. We will use the term \emph{unitary process} for such cases\footnote{Ref.~\cite{Araujo2017purification} used the term ``pure process'' for what we call here unitary process.}.

\section{Superpositions of causal orders}
We are now in a position to discuss superpositions of differently ordered processes. Let us first analyse the quantum switch, to see what type of superposition it represents.

There are two different versions relevant to the present discussion. The first is a tripartite process, where $A^1\equiv A$ and $A^2\equiv B$ act on a target system, while a third party, $A^3\equiv F$ (for ``future''), receives the target system resulting from $A$ and $B$'s operations, as well as a control system. In this version, the initial state of the target is set to some fixed state $\ket{\psi}$, while the control is fixed to $\ket{+}=\frac{1}{\sqrt{2}}\left(\ket{0}+\ket{1}\right)$. The resulting process vector is \cite{araujo15}
\begin{multline}
\ket{\omega^{(3)}_{\textrm{switch}}}= \frac{1}{\sqrt{2}}\left(\ket{0}^{F_{\textrm{c}}}\ket{\psi}^{A_I}\Ket{\id}^{A_OB_I}\Ket{\id}^{B_O F_{\textrm{t}}} \right. \\
\left.
+ 
\ket{1}^{F_{\textrm{c}}}\ket{\psi}^{B_I}\Ket{\id}^{B_OA_I}\Ket{\id}^{A_O F_{\textrm{t}}}
\right),
\label{switch3}
\end{multline}
where $F_{\textrm{c}}$ and $F_{\textrm{t}}$ respectively denote the control and target subsystems of $F$. Here, $F$ has only input space, while its output is trivial, as no influence from $F$ to $A$ and $B$ is possible.

Expression \eqref{switch3} manifestly represents a superposition of two differently ordered processes, with the target system going first to $A$ and then $B$ or vice versa.
However, seen as an isometry, the process maps the two-qubit space ($A_OB_O$) to the larger four-qubit space ($A_IB_IF_{\textrm{c}}F_{\textrm{t}}$). This means that the isometry does not describe a `pure transformation', but also the specification of a state. The natural interpretation is that process \eqref{switch3} displays \textit{entanglement} between the control system $F_{\textrm{c}}$ and the causal order of $A$ and $B$; therefore, it does not represent a superposition of causal orders alone. Indeed, if we trace out the control system, the resulting reduced process is causally separable (i.e., a classical mixture of causally ordered ones)~\cite{araujo15, oreshkov15}. 

The second version of the switch features an extra laboratory, $A^0\equiv P$ (for `past'), where both the control and target systems can be prepared (corresponding to subsystems $P_{\textrm{c}}$, $P_{\textrm{t}}$, respectively). As $P$ does not receive any state, its input space can be taken to be trivial. The resulting unitary process is
\begin{multline}
\ket{\omega^{(4)}_{\textrm{switch}}}= \left(\ket{0}^{P_{\textrm{c}}}\ket{0}^{F_{\textrm{c}}}\Ket{\id}^{P_{\textrm{t}}A_I}\Ket{\id}^{A_OB_I}\Ket{\id}^{B_O F_{\textrm{t}}} \right. \\
\left.
+ 
\ket{1}^{P_{\textrm{c}}}\ket{1}^{F_{\textrm{c}}}\Ket{\id}^{P_{\textrm{t}}B_I}\Ket{\id}^{B_OA_I}\Ket{\id}^{A_O F_{\textrm{t}}}
\right).
\label{switch4}
\end{multline}
This process is the sum of two vectors, each responsible for a different causal order. However, the two vectors do not define individually valid processes, as they are not normalised for arbitrary CPTP maps. 
Therefore, process \eqref{switch4} does not represent a superposition of processes with different causal orders, but rather describes \emph{quantum control} of the order of $A$ and $B$. 

As mentioned earlier, one can obtain a superposition of valid processes through postselection: By preparing $P_\mathrm{c}$ in state $\ket{+}$ and postselecting \textbf{$F_\mathrm{c}$} in the same state, process \eqref{switch4} reduces to 
\begin{multline}
\ket{\omega_{\mathrm{post}}}= \frac{1}{2}\left(\Ket{\id}^{P_{\textrm{t}}A_I}\Ket{\id}^{A_OB_I}\Ket{\id}^{B_O F_{\textrm{t}}} \right. \\
\left.
+ 
\Ket{\id}^{P_{\textrm{t}}B_I}\Ket{\id}^{B_OA_I}\Ket{\id}^{A_O F_{\textrm{t}}}
\right).
\label{postselected}
\end{multline}
This, however, is not a valid process vector (nor is it proportional to one), meaning that the probability to get the desired outcome $\ket{+}$ depends on the operations performed at $A$, $B$ (and it can even vanish, if the operations anticommute). Therefore, $\ket{\omega_{\mathrm{post}}}$ is not the  superposition of orders we were looking for either.

The above observations lead us to the following:

\begin{definition}
A process vector $\ket{\omega}$ for parties $P,A^1,\dots,A^n,F$ (where $P$ and $F$ have trivial input and output space, respectively) is a \emph{pure superposition of two causal orders} if 
\begin{itemize}
	\item $\ket{\omega}$ represents a unitary process;
	\item there are two nonvanishing 
complex numbers $\alpha$, $\beta$ and two differently ordered, unitary process vectors $\ket{\omega_1}$, $\ket{\omega_2}$ for the same set of parties as $\ket{\omega}$ such that
\begin{equation}
\ket{\omega} = \alpha \ket{\omega_1} + \beta \ket{\omega_2}.
\label{superpositionorders}
\end{equation}
\end{itemize}
\end{definition}

The presence of a past ($P$) and future ($F$) party in the definition follows from the fact that, in a causally ordered process, there must be an initial party, which cannot receive signals from any other, and a final party, which cannot signal to the rest. If the process is unitary, such parties must have trivial input and output space, respectively. As such parties can only be respectively first and last in any causally ordered process, the simplest candidate superposition of orders has to be four-partite, as in the unitary switch \eqref{switch4}.

Although we do not know if a superposition of causal orders is possible in general, we can rule it out for the broad class of \emph{Markovian} processes. These describe evolution from each region to the next, without any memory carried over through an external environment across different steps. A unitary, Markovian process is simply a sequence of unitaries, as in the example discussed in Sec.~\ref{processsection}. 

A Markovian process matrix is the tensor product of processes describing the individual time evolutions from each region to the next \cite{costa2016, Pollock2018, giarmatzi2018witnessing}. Therefore, a unitary, Markovian process matrix compatible with the causal order $\sigma$ has the form
\begin{equation}
\label{unitarymarkovian}
\ket{\omega}
= \bigotimes_{j=0}^{n} \Ket{U_j}^{A^{\sigma(j)}_O A^{\sigma(j+1)}_I},
\end{equation}
where 
$U_j$ are unitary matrices, we identify $A_O^0\equiv P$, $A_I^{n+1}\equiv F$, and we assume that the permutation $\sigma$ leaves first and last parties unchanged,
 $\sigma(0)=0$, $\sigma(n+1)=n+1$.

Note that, for such a process, each party can always signal to any future one (conditioned on the intermediate parties performing appropriate operations). Therefore, a process of this form compatible with $\sigma$ is necessarily incompatible with any $\sigma'\neq \sigma$.

Our core result is that, for the simplest four-partite case,
a linear combination of two unitary, Markovian processes with different orders cannot satisfy the process vector normalisation constraints, Eq.~\eqref{pureconstraints}. 
In particular, a CPTP map for the output-only space $P$ is an arbitrary density matrix $\rho^P\geq 0$, with $\tr\rho^P=1$ (describing a state preparation), while the only CPTP map for the input-only space $F$ is $\id^{F}$ (describing an arbitrary POVM measurement whose outcome is ignored). For $X=A, B$, a CPTP map is represented by a matrix $\xi^{X_IX_O} \geq 0$ such that $\tr_{X_O}\xi^{X_IX_O}=\id^{X_I}$. The technical statement of our result, proved in the Appendix\footnote{While completing the current manuscript, the author was made aware of related results in Ref.~\cite{Yokojima2020}. In particular, Lemma \ref{lemma} can be deduced from Corollary 5 in Ref.~\cite{Yokojima2020}, although the consequence as a no go theorem for superpositions of orders is not discussed there. The argument and the proof in the present manuscript were previously presented at \cite{manitoulin2017}.}, is then as follows:
\begin{lemma}
\label{lemma}
For every set of unitaries $\left\{U_j, V_j\right\}_{j=1,2,3}$, and complex numbers $\alpha, \beta \neq 0$, it is possible to find matrices $\rho$, $\xi$, $\eta\geq 0$, with $\tr\rho=1$, $\tr_{A_O}\xi^{A_IA_O}=\id^{A_I}$, and $\tr_{B_O}\eta^{B_IB_O}=\id^{B_I}$, such that 
\begin{equation}\label{thesis}
\bra{\omega} \left(\rho^{P}\otimes\xi^{A_IA_O}\otimes\eta^{B_IB_O}\otimes\id^{F}\right)\ket{\omega} \neq 1,
\end{equation}
where $\ket{\omega}$ is defined by Eq.~\eqref{superpositionorders} and
\begin{align}
\label{AB}
\ket{\omega_1} &= \Ket{U_1}^{PA_I}\Ket{U_2}^{A_OB_I}\Ket{U_3}^{B_OF}, \\
\label{BA} 
\ket{\omega_2} &= \Ket{V_1}^{PB_I}\Ket{V_2}^{B_OA_I}\Ket{V_3}^{A_OF}.
\end{align}
\end{lemma}
This result rules out the simplest, most intuitive superpositions of causal orders. In particular, it proves that the control system is necessary for the quantum switch to be a valid process. Note that, because each step is unitary, all input and output spaces must have equal dimension\footnote{The unitaries in Eq.~\eqref{AB} imply $d^P=d^{A_I}$, $d^{A_O}=d^{B_I}$, and $d^{B_O}=d^{F}$, while from Eq.~\eqref{BA} we get $d^P=d^{B_I}$, $d^{B_O}=d^{A_I}$, and $d^{A_O}=d^{F}$. Taken together, the two sets of equations imply that all dimensions are the same.}. This is not guaranteed for arbitrarily many parties. As it turns out, we need the additional assumption of equal dimensions to prove our general result:
\begin{theorem}
Consider a set of parties $P,A^1,\dots,A^n,F$, where $P$ has trivial input, $F$ has trivial output, and all nontrivial input and output dimensions are equal. For every pair of differently ordered, unitary, Markovian process vectors $\ket{\omega_1}$, $\ket{\omega_2}$, and complex numbers $\alpha, \beta \neq 0$, the linear combination $\ket{\omega} = \alpha \ket{\omega_1} + \beta \ket{\omega_2}$ is not a valid process vector.
\end{theorem}
\begin{proof}
We can prove the statement by reducing the general case to the $n=2$ one. Indeed, let $\sigma_1\neq \sigma_2$ be the permutations defining the causal orders of $\ket{\omega_1}$ and $\ket{\omega_2}$. As the two permutations are different, we can find $j\neq k$ such that $\sigma_1(j)< \sigma_1(k)$ and $\sigma_2(k)< \sigma_2(j)$. Inserting identity unitaries in all regions except $j$, $k$, we obtain the bipartite reduced processes \cite{araujo15} $\ket{\tilde{\omega}_{\sharp}}:= \left(\bigotimes_{i\neq j,k}\Bra{\id}^{A^i_IA^i_O} \right)\ket{\omega_{\sharp}}$, with $\ket{\omega_{\sharp}}=\ket{\omega}, \ket{\omega_1},\ket{\omega_2}$. Now, we have $\ket{\tilde{\omega}} = \alpha \ket{\tilde{\omega}_{1}} + \beta \ket{\tilde{\omega}_{2}}$, where $\ket{\tilde{\omega}_{1}}$ and $\ket{\tilde{\omega}_{2}}$ are unitary, Markovian, and differently ordered. Lemma \ref{lemma} then implies that $\ket{\tilde{\omega}}$ cannot be a valid process, which in turn implies that $\ket{\omega}$ cannot be a valid process either.
\end{proof}

\section{Conclusions}
We have seen that superpositions of processes with different causal orders are strongly constrained, if we require that the order of events is not correlated with any additional system. In particular, it is not possible to superpose sequences of equal-dimension unitaries connecting events in different orders---arguably the most natural candidate superposition of causal orders. Relaxing the assumptions in the no-go theorem leads to some open possibilities: it is currently unclear whether it is possible to superpose unitary, differently ordered processes that do not satisfy Markovianity or where the local input and output spaces do not have equal dimensions. Another interesting direction is to consider superpositions of more than two processes. Furthermore, it is an intriguing possibility to develop a notion of `coherence' of causal orders that applies to nonpure processes, possibly generalising corresponding resource theories for states \cite{Theurer2017, Bischof2019}. 

\begin{acknowledgments}
The author thanks Mateus Ara{\'u}jo, Cyril Branciard, {\v C}aslav Brukner, Matthew Palmer, and Magdalena Zych for useful discussions. This work was partially supported through an Australian Research Council Discovery Early Career Researcher Award (DE170100712). The University of Queensland (UQ) acknowledges the Traditional Owners and their custodianship of the lands on which UQ operates.
\end{acknowledgments}

\small


\providecommand{\href}[2]{#2}\begingroup\raggedright\endgroup


\normalsize

\onecolumn\newpage
\appendix

\section{Proof of Lemma \ref{lemma}}
Here we provide the proof of Lemma \ref{lemma}, which we reformulate for convenience:
\begin{lemma*}
For every set of unitaries $\left\{U_j, V_j\right\}_{j=1,2,3}$, and complex numbers $\alpha, \beta \neq 0$, it is possible to find
\begin{align}\label{locals}
\chi:=& \rho^{P}\otimes\xi^{A_IA_O}\otimes\eta^{B_IB_O}\otimes\id^{F},\quad\textrm{with}\\\nonumber
\tr \rho^{P}=&1, \; \tr_{A_O}\xi^{A_IA_O}=\id^{A_I}, \; \tr_{B_O}\eta^{B_IB_O}=\id^{B_I},
\end{align} 
such that 
\begin{equation}\label{thesisA}
\langle \chi \rangle := \bra{\omega} \chi \ket{\omega} \neq 1,
\end{equation}
where $\ket{\omega}=\alpha\ket{\omega_1}+\beta\ket{\omega_2}$ and
\begin{align}
\label{ABA}
\ket{\omega_1} &= \Ket{U_1}^{PA_I}\Ket{U_2}^{A_OB_I}\Ket{U_3}^{B_OF}, \\
\label{BAA} 
\ket{\omega_2} &= \Ket{V_1}^{PB_I}\Ket{V_2}^{B_OA_I}\Ket{V_3}^{A_OF}.
\end{align}
\end{lemma*}
As $\ket{\omega_1}$ and $\ket{\omega_2}$ are normalised process vectors, we have $\bra{\omega_1} \chi \ket{\omega_1} = \bra{\omega_2} \chi \ket{\omega_2} =1$ for all $\chi$ of the form \eqref{locals}, so that
\begin{equation}
\langle \chi \rangle = \left|\alpha\right|^2 +  \left|\beta\right|^2 + 2\mathrm{Re}\left(\alpha^*\beta \bra{\omega_1}\chi\ket{\omega_2} \right).
\end{equation}
For $\ket{\omega}$ to be a normalised process, we need $\langle \chi \rangle = 1$ for all $\chi$ of the form \eqref{locals}, which requires $\mathrm{Re}\left(\alpha^*\beta \bra{\omega_1}\chi\ket{\omega_2}\right)$ to be constant over $\chi$. Writing $\alpha^*\beta = \ee^{\ii\phi} \left|\alpha^*\beta \right|$, we obtain that the cross term $\mathrm{Re}\left(\ee^{\ii\phi}\bra{\omega_1}\chi\ket{\omega_2}\right)$ should be constant for all $\chi$.

In order to simplify the cross term, we restrict the CPTP maps to a pure state preparation $\ket{\psi}$ for $P$ and unitaries $R$, $S$ for $A$ and $B$, respectively, so that
\begin{equation}
\chi = \proj{\psi^*}^{P}\otimes \Proj{R^*}^{A_IA_O}\otimes \Proj{S^*}^{B_IB_O}\otimes \id^{F},
\label{unitaries}
\end{equation}
with the notation $\Proj{A}\equiv \KetBra{A}$ and where the complex conjugations in $\ket{\psi^*}$, $R^*$, $S^*$, result from the transposition in our definition of the Choi representation for local operations, Eq.~\eqref{Choi}. After some linear algebra, we arrive at 
\begin{equation}
\bra{\omega_1}\chi\ket{\omega_2} = \bra{\psi}U_1^{\dag}R^{\dag}U_2^{\dag}S^{\dag}U_3^{\dag}V_3 R V_2 S V_1 \ket{\psi}.
\label{complicated}
\end{equation} 
The crucial step of the proof is to simplify this expression through an appropriate choice of local operations. We can do this using the substitutions
\begin{align} \label{subs1}
S\mapsto S' &= U_3^{\dag}V_3U_2^{\dag}V_1SV_1^{\dag}, \\  \label{subs2}
R\mapsto R' &= U_2^{\dag}V_1RV_1^{\dag}U_2V_3^{\dag}U_3V_2^{\dag},
\end{align}
under which the cross term takes the form
\begin{align} 
\label{simpler}
\ee^{\ii\phi}&\bra{\omega_1}  \chi\ket{\omega_2} = \bra{\psi}V T \ket{\psi}, \\ \label{V}
V &= \ee^{\ii\phi} U_1^{\dag}V_2 U_3^{\dag}V_3U_2^{\dag}V_1  , \\ 
\label{T} T&= R^{\dag}S^{\dag} R S.
\end{align}

Now the normalisation of $\ket{\omega}$ reduces to the condition that, given $V$ as in Eq.~\eqref{V}, $\mathrm{Re}\left( \bra{\psi}V T \ket{\psi} \right)$ has to be constant for all unitaries $T$ of the form \eqref{T} and all $\ket{\psi}$. This implies that  $VT + T^{\dag}V^{\dag}$ should be proportional to the identity or, equivalently,
\begin{equation}
TVT + V^{\dag} = \lambda T,
\label{condition}
\end{equation}
where $\lambda$ is a number that depends on $\alpha$ and $\beta$:
\begin{equation}
\lambda = \frac{1-\left|\alpha\right|^2-\left|\beta\right|^2}{2\left|\alpha^*\beta\right|}.
\label{lambda}
\end{equation}
To show that no unitary $V$ can satisfy condition \eqref{condition} for all $R$, $S$, it is sufficient to take block-diagonal unitaries $R$, $S$ of the form $\sigma\oplus\id$, where $\sigma$ acts on a two-dimensional subspace spanned by eigenstates of $V$ (so that the restriction of $V$ to that subspace is also unitary). In other words, it is sufficient to show that condition \eqref{condition} cannot be satisfied by any single-qubit unitary $V$. We will omit the $\oplus\id$ for notational convenience.

Let us consider first the case where $\lambda=0$; that is, a superposition of processes with normalised amplitudes $\left|\alpha\right|^2+\left|\beta\right|^2=1$. Condition \eqref{condition} reduces to
\begin{equation}
TVT + V^{\dag} = 0,
\label{simplercondition}
\end{equation}
which has to hold for all $T$ of the form \eqref{T}. We have to find a set of unitaries $T$ of that form such that a single unitary $V$ cannot satisfy \reftoeq{simplercondition} for all of them. First, choose $R=S=\id$ to set $T=\id$, so that \reftoeq{simplercondition} gives 
\begin{equation}
V + V^{\dag} = 0.
\label{antiself}
\end{equation}
Substituting this into condition \eqref{simplercondition}, we now have
\begin{equation}
TVT = V.
\label{TVT}
\end{equation}
We can see that this is not possible by choosing $R_j$, $S_j$ such that $\ii\sigma_j =T_j= R_j^{\dag}S_j^{\dag} R_j S_j$, with $j=x,y,z$ labelling the three Pauli matrices. (For example, for $j=x$, we can choose $R_x=\sigma_y$, $S_x=\frac{\id+\ii \sigma_x}{\sqrt{2}}$\footnote{The general construction is $R_j=\sigma_k$ (with $k\neq j$) and $S_j=\frac{\id + \ii \sigma_j}{\sqrt{2}}$. This gives $R_jS_j =\frac{\sigma_k - \sum_l\epsilon_{kjl}\sigma_l}{\sqrt{2}} $ and $R^{\dag}_jS^{\dag}_j =\frac{\sigma_k + \sum_l\epsilon_{kjl}\sigma_l}{\sqrt{2}}$. As $(\sigma_k\sigma_j)^2 = -\id$ for $k\neq j$, we get $R_j^{\dag}S_j^{\dag} R_j S_j = \frac{\ii}{2}\left(\sigma_j - \sigma_k\sigma_j\sigma_k\right) = \ii \sigma_j$.}.) Now, $V$ has to satisfy
\begin{equation}
\sigma_j V \sigma_j = -V, \qquad j=x,y,z.
\end{equation}
This can only hold for $V=0$, which is not a unitary. 

Let us now go back to the $\lambda\neq 0$ case. As we have seen above, we can choose $R$ and $S$ such that $T= \ii \sigma_x$. But we can also set $T=-\ii \sigma_x$ by choosing $R=\sigma_y$, $S=\frac{\id-\ii \sigma_x}{\sqrt{2}}$. Thus, we get the two equations
\begin{align}
-\sigma_x & V \sigma_x + V^{\dag} =  i\lambda \sigma_x,\\
-\sigma_x & V \sigma_x + V^{\dag} = - i\lambda \sigma_x.
\end{align}
These can be satisfied simultaneously only for $\lambda=0$, which brings us back to the previous case, concluding the proof. \qed


\begin{thebibliography}{10}

\bibitem{Feynman1948}
R.~P. Feynman, ``Space-Time Approach to Non-Relativistic Quantum Mechanics,''
  \href{http://dx.doi.org/10.1103/RevModPhys.20.367}{{\em Rev. Mod. Phys.}
  {\bfseries 20}, 367--387 (1948)}.

\bibitem{Butterfield2001}
J.~Butterfield and C.~Isham, {\em Spacetime and the philosophical challenge of
  quantum gravity},
  \href{http://dx.doi.org/10.1017/CBO9780511612909.003}{p.~33–89}.
\newblock Cambridge University Press.
\newblock \href{http://arxiv.org/abs/arXiv:gr-qc/9903072}{{\ttfamily
  arXiv:gr-qc/9903072}}.

\bibitem{hardy2007towards}
L.~{Hardy}, ``{Towards quantum gravity: a framework for probabilistic theories
  with non-fixed causal structure},''
  \href{http://dx.doi.org/10.1088/1751-8113/40/12/S12}{{\em J.~Phys. A: Math.
  Gen.} {\bfseries 40}, 3081--3099 (2007)}.

\bibitem{Zych2019}
M.~Zych, F.~Costa, I.~Pikovski, and {\v{C}}.~Brukner, ``Bell's theorem for
  temporal order,'' \href{http://dx.doi.org/10.1038/s41467-019-11579-x}{{\em
  Nat. Commun.} {\bfseries 10}, 3772 (2019)}.

\bibitem{Procopio2015}
L.~M. Procopio, A.~Moqanaki, M.~Ara{\'u}jo, F.~Costa, I.~A. Calafell, E.~G.
  Dowd, D.~R. Hamel, L.~A. Rozema, {\v C}.~Brukner, and P.~Walther,
  ``Experimental superposition of orders of quantum gates,''
  \href{http://dx.doi.org/10.1038/ncomms8913}{{\em Nat. Commun.} {\bfseries 6},
  7913 (2015)}.

\bibitem{Rubinoe1602589}
G.~Rubino, L.~A. Rozema, A.~Feix, M.~Ara{\'u}jo, J.~M. Zeuner, L.~M. Procopio,
  {\v C}.~Brukner, and P.~Walther, ``Experimental verification of an indefinite
  causal order,'' \href{http://dx.doi.org/10.1126/sciadv.1602589}{{\em Sci.
  Adv.} {\bfseries 3}, e1602589 (2017)}.

\bibitem{rubino2017experimental}
G.~Rubino, L.~A. Rozema, F.~Massa, M.~Ara{\'{u}}jo, M.~Zych, {\v{C}}.~Brukner,
  and P.~Walther, ``Experimental entanglement of temporal order,''
  \href{http://dx.doi.org/10.22331/q-2022-01-11-621}{{\em {Quantum}} {\bfseries
  6}, 621 (2022)}.

\bibitem{Goswami2018}
K.~Goswami, C.~Giarmatzi, M.~Kewming, F.~Costa, C.~Branciard, J.~Romero, and
  A.~G. White, ``Indefinite Causal Order in a Quantum Switch,''
  \href{http://dx.doi.org/10.1103/PhysRevLett.121.090503}{{\em Phys. Rev.
  Lett.} {\bfseries 121}, 090503 (2018)}.

\bibitem{goswami2018communicating}
K.~Goswami, Y.~Cao, G.~A. Paz-Silva, J.~Romero, and A.~G. White, ``Increasing
  communication capacity via superposition of order,''
  \href{http://dx.doi.org/10.1103/PhysRevResearch.2.033292}{{\em Phys. Rev.
  Research} {\bfseries 2}, 033292 (2020)}.

\bibitem{guo2018experimental}
Y.~Guo, X.-M. Hu, Z.-B. Hou, H.~Cao, J.-M. Cui, B.-H. Liu, Y.-F. Huang, C.-F.
  Li, G.-C. Guo, and G.~Chiribella, ``Experimental Transmission of Quantum
  Information Using a Superposition of Causal Orders,''
  \href{http://dx.doi.org/10.1103/PhysRevLett.124.030502}{{\em Phys. Rev.
  Lett.} {\bfseries 124}, 030502 (2020)}.

\bibitem{Wei2019}
K.~Wei, N.~Tischler, S.-R. Zhao, Y.-H. Li, J.~M. Arrazola, Y.~Liu, W.~Zhang,
  H.~Li, L.~You, Z.~Wang, Y.-A. Chen, B.~C. Sanders, Q.~Zhang, G.~J. Pryde,
  F.~Xu, and J.-W. Pan, ``Experimental Quantum Switching for Exponentially
  Superior Quantum Communication Complexity,''
  \href{http://dx.doi.org/10.1103/PhysRevLett.122.120504}{{\em Phys. Rev.
  Lett.} {\bfseries 122}, 120504 (2019)}.

\bibitem{taddei2020experimental}
M.~M. Taddei, J.~Cari\~ne, D.~Mart\'{\i}nez, T.~Garc\'{\i}a, N.~Guerrero, A.~A.
  Abbott, M.~Ara\'ujo, C.~Branciard, E.~S. G\'omez, S.~P. Walborn, L.~Aolita,
  and G.~Lima, ``Computational Advantage from the Quantum Superposition of
  Multiple Temporal Orders of Photonic Gates,''
  \href{http://dx.doi.org/10.1103/PRXQuantum.2.010320}{{\em PRX Quantum}
  {\bfseries 2}, 010320 (2021)}.

\bibitem{chiribella09}
G.~{Chiribella}, G.~M. {D'Ariano}, P.~{Perinotti}, and B.~{Valiron}, ``{Quantum
  computations without definite causal structure},''
  \href{http://dx.doi.org/10.1103/PhysRevA.88.022318}{{\em Phys. Rev.~A}
  {\bfseries 88}, 022318 (2013)}.

\bibitem{chiribella12}
G.~{Chiribella}, ``{Perfect discrimination of no-signalling channels via
  quantum superposition of causal structures},''
  \href{http://dx.doi.org/10.1103/PhysRevA.86.040301}{{\em Phys. Rev.~A}
  {\bfseries 86}, 040301(R) (2012)}.

\bibitem{colnaghi11}
T.~{Colnaghi}, G.~M. {D'Ariano}, S.~{Facchini}, and P.~{Perinotti}, ``{Quantum
  computation with programmable connections between gates},''
  \href{http://dx.doi.org/10.1016/j.physleta.2012.08.028}{{\em Phys. Lett.~A}
  {\bfseries 376}, 2940--2943 (2012)}.

\bibitem{araujo14}
M.~{Ara{\'u}jo}, F.~{Costa}, and {\v C}.~{Brukner}, ``{Computational Advantage
  from Quantum-Controlled Ordering of Gates},''
  \href{http://dx.doi.org/10.1103/PhysRevLett.113.250402}{{\em Phys. Rev.
  Lett.} {\bfseries 113}, 250402 (2014)}.

\bibitem{feixquantum2015}
A.~Feix, M.~Ara\'ujo, and {\v C}.~Brukner, ``Quantum superposition of the order
  of parties as a communication resource,''
  \href{http://dx.doi.org/10.1103/PhysRevA.92.052326}{{\em Phys. Rev. A}
  {\bfseries 92}, 052326 (2015)}.

\bibitem{Guerin2016}
P.~A. Gu\'erin, A.~Feix, M.~Ara\'ujo, and {\v C}.~Brukner, ``Exponential
  Communication Complexity Advantage from Quantum Superposition of the
  Direction of Communication,''
  \href{http://dx.doi.org/10.1103/PhysRevLett.117.100502}{{\em Phys. Rev.
  Lett.} {\bfseries 117}, 100502 (2016)}.

\bibitem{Ebler2018}
D.~Ebler, S.~Salek, and G.~Chiribella, ``Enhanced Communication with the
  Assistance of Indefinite Causal Order,''
  \href{http://dx.doi.org/10.1103/physrevlett.120.120502}{{\em Phys. Rev.
  Lett.} {\bfseries 120}, 120502 (2018)}.

\bibitem{Salek2018}
S.~Salek, D.~Ebler, and G.~Chiribella, ``Quantum communication in a
  superposition of causal orders,''
  \href{http://arxiv.org/abs/1809.06655v2}{{\ttfamily arXiv:1809.06655v2
  [quant-ph]}}.

\bibitem{Gupta2019}
M.~K. Gupta and U.~Sen, ``Transmitting quantum information by superposing
  causal order of mutually unbiased measurements,''
  \href{http://arxiv.org/abs/arXiv:1909.13125v1}{{\ttfamily arXiv:1909.13125v1
  [quant-ph]}}.

\bibitem{Aharonov1990}
Y.~Aharonov, J.~Anandan, S.~Popescu, and L.~Vaidman, ``Superpositions of time
  evolutions of a quantum system and a quantum time-translation machine,''
  \href{http://dx.doi.org/10.1103/PhysRevLett.64.2965}{{\em Phys. Rev. Lett.}
  {\bfseries 64}, 2965 (1990)}.

\bibitem{oreshkov12}
O.~{Oreshkov}, F.~{Costa}, and {\v C}.~{Brukner}, ``{Quantum correlations with
  no causal order},'' \href{http://dx.doi.org/10.1038/ncomms2076}{{\em Nat.
  Commun.} {\bfseries 3}, 1092 (2012)}.

\bibitem{Oeckl2003318}
R.~Oeckl, ``A “general boundary” formulation for quantum mechanics and
  quantum gravity,''
  \href{http://dx.doi.org/http://dx.doi.org/10.1016/j.physletb.2003.08.043}{{\em
  Phys. Lett. B} {\bfseries 575}, 318--324 (2003)}.

\bibitem{Aharonov2009}
Y.~Aharonov, S.~Popescu, J.~Tollaksen, and L.~Vaidman, ``Multiple-time states
  and multiple-time measurements in quantum mechanics,''
  \href{http://dx.doi.org/10.1103/PhysRevA.79.052110}{{\em Phys.\ Rev.\ A}
  {\bfseries 79}, 052110 (2009)}.

\bibitem{Cotler2017}
J.~Cotler, C.-M. Jian, X.-L. Qi, and F.~Wilczek, ``Superdensity operators for
  spacetime quantum mechanics,''
  \href{http://dx.doi.org/10.1007/jhep09(2018)093}{{\em J. High Energ. Phys.}
  {\bfseries 2018}, 93 (2018)}.

\bibitem{Silva2017}
R.~Silva, Y.~Guryanova, A.~J. Short, P.~Skrzypczyk, N.~Brunner, and S.~Popescu,
  ``{Connecting processes with indefinite causal order and multi-time quantum
  states},'' \href{http://dx.doi.org/10.1088/1367-2630/aa84fe}{{\em New J.\
  Phys.} {\bfseries 19}, 103022 (2017)}.

\bibitem{Barrett2019}
J.~Barrett, R.~Lorenz, and O.~Oreshkov, ``Quantum Causal Models,''
  \href{http://arxiv.org/abs/1906.10726v2}{{\ttfamily arXiv:1906.10726v2
  [quant-ph]}}.

\bibitem{Guerin2018}
P.~A. Gu{\'{e}}rin and {\v{C}}.~Brukner, ``Observer-dependent locality of
  quantum events,'' \href{http://dx.doi.org/10.1088/1367-2630/aae742}{{\em New
  J. Phys.} {\bfseries 20}, 103031 (2018)}.

\bibitem{Oreshkov2019timedelocalized}
O.~Oreshkov, ``Time-delocalized quantum subsystems and operations: on the
  existence of processes with indefinite causal structure in quantum
  mechanics,'' \href{http://dx.doi.org/10.22331/q-2019-12-02-206}{{\em
  {Quantum}} {\bfseries 3}, 206 (2019)}.

\bibitem{Heinosaari2011}
T.~Heinosaari and M.~Ziman,
  \href{http://dx.doi.org/10.1017/CBO9781139031103}{{\em The Mathematical
  Language of Quantum Theory: From Uncertainty to Entanglement}}.
\newblock Cambridge University Press, 2011.

\bibitem{Choi1975}
M.-D. Choi, ``Completely positive linear maps on complex matrices,''
  \href{http://dx.doi.org/10.1016/0024-3795(75)90075-0}{{\em Linear Algebra
  Appl.} {\bfseries 10}, 285--290 (1975)}.

\bibitem{jamio72}
A.~Jamiołkowski, ``Linear transformations which preserve trace and positive
  semidefiniteness of operators,''
  \href{http://dx.doi.org/10.1016/0034-4877(72)90011-0}{{\em Rep. Math. Phys}
  {\bfseries 3}, 275--278 (1972)}.

\bibitem{Shrapnel2017}
S.~Shrapnel, F.~Costa, and G.~Milburn, ``Updating the Born rule,''
  \href{http://dx.doi.org/10.1088/1367-2630/aabe12}{{\em New J. Phys.}
  {\bfseries 20}, 053010 (2018)}.

\bibitem{oreshkov2016}
O.~Oreshkov and N.~J. Cerf, ``Operational quantum theory without predefined
  time,'' \href{http://dx.doi.org/10.1088/1367-2630/18/7/073037}{{\em New J.
  Phys.} {\bfseries 18}, 073037 (2016)}.

\bibitem{Milz2018}
S.~Milz, F.~A. Pollock, T.~P. Le, G.~Chiribella, and K.~Modi, ``Entanglement,
  non-Markovianity, and causal non-separability,''
  \href{http://dx.doi.org/10.1088/1367-2630/aaafee}{{\em New J.~Phys.}
  {\bfseries 20}, 033033 (2018)}.

\bibitem{oreshkov15}
O.~Oreshkov and C.~Giarmatzi, ``Causal and causally separable processes,''
  \href{http://dx.doi.org/10.1088/1367-2630/18/9/093020}{{\em New J. of Phys.}
  {\bfseries 18}, 093020 (2016)}.

\bibitem{Wechs2019}
J.~Wechs, A.~A. Abbott, and C.~Branciard, ``On the definition and
  characterisation of multipartite causal (non)separability,''
  \href{http://dx.doi.org/10.1088/1367-2630/aaf352}{{\em New J. of Phys.}
  {\bfseries 21}, 013027 (2019)}.

\bibitem{baumeler14}
{\"A}.~Baumeler, A.~Feix, and S.~Wolf, ``{Maximal incompatibility of locally
  classical behavior and global causal order in multi-party scenarios},''
  \href{http://dx.doi.org/10.1103/PhysRevA.90.042106}{{\em Phys. Rev. A}
  {\bfseries 90}, 042106 (2014)}.

\bibitem{baumeler2017}
{\"{A}}.~Baumeler and S.~Wolf, ``{The space of logically consistent classical
  processes without causal order},''
  \href{http://dx.doi.org/10.1088/1367-2630/18/1/013036}{{\em New J. of Phys.}
  {\bfseries 18}, 013036 (2016)}.

\bibitem{Baumeler2019}
{\"{A}}.~Baumeler, F.~Costa, T.~C. Ralph, S.~Wolf, and M.~Zych, ``Reversible
  time travel with freedom of choice,''
  \href{http://dx.doi.org/10.1088/1361-6382/ab4973}{{\em Class. Quantum Grav.}
  {\bfseries 36}, 224002 (2019)}.

\bibitem{Tobar2020}
G.~Tobar and F.~Costa, ``Reversible dynamics with closed time-like curves and
  freedom of choice,'' \href{http://dx.doi.org/10.1088/1361-6382/aba4bc}{{\em
  Classical and Quantum Gravity} {\bfseries 37}, 205011 (2020)}.

\bibitem{araujo15}
M.~Ara{\'u}jo, C.~Branciard, F.~Costa, A.~Feix, C.~Giarmatzi, and {\v
  C}.~Brukner, ``Witnessing causal nonseparability,''
  \href{http://dx.doi.org/10.1088/1367-2630/17/10/102001}{{\em New J. Phys.}
  {\bfseries 17}, 102001 (2015)}.

\bibitem{Araujo2017purification}
M.~Ara{\'{u}}jo, A.~Feix, M.~Navascu{\'{e}}s, and {\v{C}}.~Brukner, ``A
  purification postulate for quantum mechanics with indefinite causal order,''
  \href{http://dx.doi.org/10.22331/q-2017-04-26-10}{{\em {Quantum}} {\bfseries
  1}, 10 (2017)}.

\bibitem{costa2016}
F.~Costa and S.~Shrapnel, ``Quantum causal modelling,''
  \href{http://dx.doi.org/https://doi.org/10.1088/1367-2630/18/6/063032}{{\em
  New J. of Phys.} {\bfseries 18}, 063032 (2016)}.

\bibitem{Pollock2018}
F.~A. Pollock, C.~Rodr{\'{\i}}guez-Rosario, T.~Frauenheim, M.~Paternostro, and
  K.~Modi, ``Operational Markov Condition for Quantum Processes,''
  \href{http://dx.doi.org/10.1103/physrevlett.120.040405}{{\em Phys. Rev.
  Lett.} {\bfseries 120}, 040405 (2018)}.

\bibitem{giarmatzi2018witnessing}
C.~Giarmatzi and F.~Costa, ``Witnessing quantum memory in non-{M}arkovian
  processes,'' \href{http://dx.doi.org/10.22331/q-2021-04-26-440}{{\em
  {Quantum}} {\bfseries 5}, 440 (2021)}.

\bibitem{Yokojima2020}
W.~Yokojima, M.~T. Quintino, A.~Soeda, and M.~Murao, ``Consequences of
  preserving reversibility in quantum superchannels,''
  \href{http://dx.doi.org/10.22331/q-2021-04-26-441}{{\em {Quantum}} {\bfseries
  5}, 441 (2021)}.

\bibitem{manitoulin2017}
F.~Costa, ``A no-go theorem for superpositions of causal order.'' \href{https://uwaterloo.ca/spacetime-information-workshop/abstracts}{Space-time
  and Information}, Manitoulin Island, Ontario, Canada, 2017.

\bibitem{Theurer2017}
T.~Theurer, N.~Killoran, D.~Egloff, and M.~B. Plenio, ``Resource Theory of
  Superposition,'' \href{http://dx.doi.org/10.1103/PhysRevLett.119.230401}{{\em
  Phys. Rev. Lett.} {\bfseries 119}, 230401 (2017)}.

\bibitem{Bischof2019}
F.~Bischof, H.~Kampermann, and D.~Bru{\ss}, ``Resource Theory of Coherence
  Based on Positive-Operator-Valued Measures,''
  \href{http://dx.doi.org/10.1103/PhysRevLett.123.110402}{{\em Phys. Rev.
  Lett.} {\bfseries 123}, 110402 (2019)}.

\end{thebibliography}
\end{document}